\documentclass[preprint,12pt]{elsarticle}

\usepackage{amsmath}
\usepackage{amsthm}
\usepackage{algorithm}
\usepackage[noend]{algpseudocode}
\usepackage{graphicx}
\usepackage{tabularx}

\newtheorem{corollary}{Corollary}
\newtheorem{theorem}{Theorem}
\newtheorem{lemma}{Lemma}
\newtheorem{fact}{Fact}
\newcommand{\Geqt}{\ensuremath{G_\text{eqt}}}
\newcommand{\Veqt}{\ensuremath{V_\text{eqt}}}
\newcommand{\Eeqt}{\ensuremath{E_\text{eqt}}}

\newtheorem{claim}{Claim}

\newcommand{\junk}{}

\journal{Theoretical Computer Science}

\begin{document}
\begin{frontmatter}



\title{Universal Coating for Programmable Matter}


\author[asu]{Zahra Derakhshandeh\corref{cor1}\fnref{fn1}}
\ead{zderakhs@asu.edu}

\author[upb]{Robert Gmyr\fnref{fn2}}
\ead{gmyr@mail.upb.de}

\author[asu]{Andr\'ea W.\ Richa\fnref{fn1}}
\ead{aricha@asu.edu}

\author[upb]{Christian Scheideler\fnref{fn2}}
\ead{scheideler@upb.de}

\author[upb]{Thim Strothmann\fnref{fn2}}
\ead{thim@mail.upb.de}

\address[asu]{Computer Science, CIDSE, Arizona State University, USA}
\address[upb]{Department of Computer Science, Paderborn University, Germany}

\cortext[cor1]{Corresponding author}

\fntext[fn1]{Supported in part by the NSF under Awards CCF-1353089 and CCF-1422603.}
\fntext[fn2]{Supported in part by DFG grant SCHE 1592/3-1.}

\begin{abstract}
The idea behind universal coating is to have a thin layer of a specific
substance covering an object of any shape so that one can measure a certain
condition (like temperature or cracks) at any spot on the surface of the
object without requiring direct access to that spot. We study the universal
coating problem in the context of {\em self-organizing programmable matter}
consisting of simple computational elements, called {\em particles},
that can establish and release bonds and can actively move in a self-organized
way.
%, following the geometric amoebot model, in order to coat an object as
%evenly as possible.
%Nature may serve as inspiration, since prominent examples
%of self-organizing coating abound, from proteins closing
%wounds, antibodies surrounding bacteria, or ants surrounding food in order to
%transport it to their nest.
Based on that matter, we present a {\em worst-case work-optimal universal coating
algorithm} that uniformly coats any object  of arbitrary shape and size that allows a uniform coating. Our particles are anonymous,
do not have any global information, have constant-size memory, and utilize
only local interactions.
%Our universal coating algorithm has several applications (such as checking if the object to be coated is convex or not) %which we also briefly discuss in the paper.
\end{abstract}


\end{frontmatter}

\section{Introduction}

Today, engineers often need to visually inspect bridges, tunnels, wind
turbines and other large civil engineering structures for defects --- a task
that is both time-consuming and costly. In the not so distant future,
{\em smart coating} technology could do the job faster and cheaper, and increase safety at
the same time. The idea behind smart coating (also coined {\em smart paint})
%AR: added comment on smart paint above
is to have a thin layer of a specific substance covering the object so that
one can measure a certain condition (like temperature or cracks) at any spot
on the surface of the object without requiring direct access to that spot.
Also in nature, smart coating occurs in various situations. Prominent examples
are proteins closing wounds, antibodies surrounding bacteria, or ants
surrounding food in order to transport it to their nest. So one can envision a
broad range of coating applications for programmable matter in the future.
%{\bf AR: I think we need
%more motivation here, specially on applications of coating within programmable
%matter: coating bridges to repair cracks, or monitor tension, self-repair the
%winds of a space craft, coat and temporary fix leaks in a nuclear reactor
%without human intervention, coat and stop internal bleeding, etc. some of it
%has been done in the paragraph below, which I took out of the related work
%section}
We intend to study coating problems in the context of self-organizing
programmable matter consisting of simple computational elements, called
particles, that can establish and release bonds and can actively move in a
self-organized way. As a basic model for these self-organizing particle
systems, we will use the geometric version of the amoebot model presented in~\cite{DNA,spaa-ba14}.
%{\bf AR: define/talk about programmable matter, self-organizing particle
%systems and the smart paint/coating vision}

\subsection{Amoebot model} \label{sec:model}

We assume that any structure the particle system can form can be represented
as a subgraph of an infinite graph $G=(V,E)$ where $V$ represents all possible
positions the particles can occupy relative to their structure, and $E$
represents all possible atomic transitions a particle can perform %in one unit of time
 as well as all places where neighboring particles can bond to each other.
%We assume that we have a graph $G = (V ,E)$ that represents the positions that
%a connected set of particles may occupy relative to their structure --- i.e.,
%$V$ represents all possible positions of a particle (relative to the other
%particles in their structure) and $E$ represents all possible transitions
%between nodes and places where neighboring particles can connect.
In the {\em geometric amoebot model}, we assume that $G = \Geqt$, where $\Geqt =
(\Veqt, \Eeqt)$ is the infinite regular triangular grid graph, see
Figure~\ref{fig:graph}(a).

\begin{figure*}
  \centering
  \includegraphics{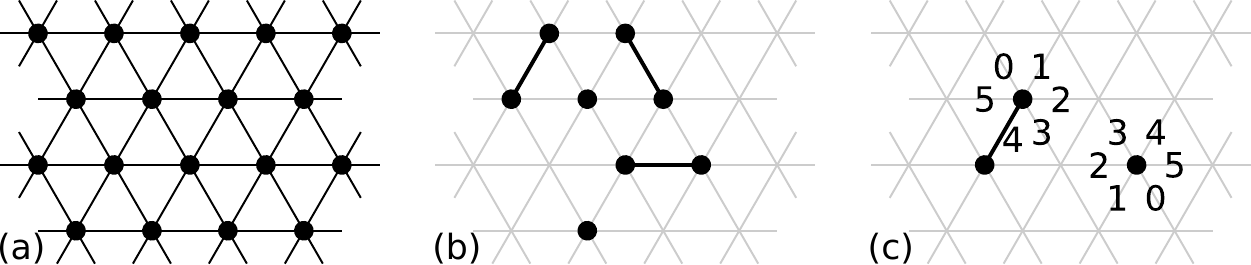}
  \vspace{-.15in}
  \caption{
    \small
    (a) shows a section of $\Geqt$. Nodes of $\Geqt$ are shown as black circles.
    (b) shows five particles on $\Geqt$. The underlying graph $\Geqt$ is depicted as a gray mesh.
    A particle occupying a single node is depicted as a black circle,
    and a particle occupying two nodes is depicted as two black circles connected by an edge.
    (c) depicts two particles occupying two non-adjacent positions on $\Geqt$.
    The particles have different offsets for their head port labelings.
  }
  \label{fig:graph}
\end{figure*}

We briefly recall the main properties of the geometric amoebot model. Each
particle occupies either a single node or a pair of adjacent nodes in $\Geqt$,
and every node can be occupied by at most one particle. Two particles
occupying adjacent nodes are \emph{connected} by a \emph{bond}, and we refer
to such particles as \emph{neighbors}. The bonds do not just ensure that the
particles form a connected structure but they are also used for exchanging
information as explained below.

Particles move through \emph{expansions} and \emph{contractions}: If a
particle occupies one node (i.e., it is \emph{contracted}), it can expand to
an unoccupied adjacent node to occupy two nodes. If a particle occupies two
nodes (i.e., it is \emph{expanded}), it can contract to one of these nodes to
occupy only a single node.  Figure~\ref{fig:graph}(b)  illustrates a set of
particles (some contracted, some expanded) on the underlying graph $\Geqt$.
For an expanded particle, we denote the node the particle last expanded into
as the \emph{head} of the particle and call the other occupied node its
\emph{tail}.
%\textbf{RG: How does a particle store this information? We only talk about distinguishable / numbered bonds later.}
%AR: moved the respective sentence later, after talking about shared memory, etc.
A \emph{handover} allows particles to stay connected as they move. Two
scenarios are possible here: (1) a contracted particle $p$ can ``push'' a
neighboring expanded particle $q$ and expand into the neighboring node
previously occupied by $q$, forcing $q$ to contract, or (2) an expanded
particle $p$ can ``pull'' a neighboring contracted particle $q$ to node $v$ it
occupies  thereby causing $q$ to expand into $v$, which allows $p$ to
contract.

Particles are \emph{anonymous} but each particle has a collection of {\em
ports}, one for each edge incident to the nodes occupied by it, that have
unique labels. Adjacent particles establish bonds through the ports facing
each other. We also assume that the particles have a common {\em chirality},
i.e., they all have the same notion of {\em clockwise (CW) direction}, which
allows each particle $p$ to order its head port labels in clockwise order.
However, particles do not have a common sense of orientation since they can
have different offsets of the labelings, see Figure~\ref{fig:graph}(c).
W.l.o.g.\footnote{Without loss of generality.}, we assume that each particle
labels its head ports from $0$ to $5$ in clockwise order. Whenever a particle
$p$ is connected to a particle $q$, we assume that $p$ knows the label of
$q$'s bond that $p$ is connected with.

Each particle has a constant-size shared local memory
%, and each port of it has a constant-size memory
that can be read and written to by any neighboring particle. This allows a
particle to exchange information with a neighboring particle by simply writing it into
the other particle's memory.\footnote{In~\cite{DNA,spaa-ba14}, the model was presented as having a shared memory for each port that is visible only to the respective neighbor: The two variants of the model are equivalent, in the sense that they can emulate each other trivially; we adopt the one here for convenience.}
%and any pair of connected particles has a constant-size shared memory,
%associated with and synchronized by the corresponding ports, that can be read
%and written to by both of them.
%Each particle can access this shared memory using the respective edge label.
A particle always knows whether it is contracted or expanded, and in the
latter case it also knows along which head port label it is expanded.
W.l.o.g.~we assume that this information is also available to the neighboring
particles (by publishing that label in its local shared memory).
% --- and this information will be available to neighboring particles.
Particles do not know the total number of particles, nor do they have any
estimate on this number.

We assume the standard asynchronous model from distributed computing, where
the particle system progresses through a sequence of {\em particle
activations}, i.e., only one particle is active at a time.
%\footnote{In
%reality, this may not be the case, but a local contention resolution rule
%could be used to make sure that only one particle within a local neighborhood
%is active at a time, which would be sufficient for our results to hold in a
%concurrent environment.}
Whenever a particle is activated, it can perform an
arbitrary bounded amount of computation (involving its local memory as well as
the shared memories of its neighbors)
and at most one movement. A {\em round} is defined as the elapsed time until each particle has been activated at least once.

 We count time according to the number of particle activations that have
already happened since the start of the activation sequence. We assume the
activation sequence to be {\em fair}, i.e., at any point in time, every
particle will eventually be activated. The \emph{configuration} $C$ of the
system at the beginning of time $t$ consists of the nodes in $\Geqt$ occupied
by the object and the set of particles; in addition, for every particle $p$, $C$ contains the current state of $p$, including
whether the particle is expanded or contracted, its port labeling, and the
contents of its local memory. The {\em work} spent by the particles till time
$t$ is measured by the number of movements they have done until that point. (We ignore
other state changes since their energy consumption should be irrelevant
compared to the energy for a movement.) For more details on the model, please
refer to~\cite{DNA}.

%Let ${\cal S}$ be the set of all system states in which the particle system is
%connected. In general, a {\em computational problem} $P$ is specified by a set
%${\cal S}' \subseteq {\cal S}$ of permitted initial system states and a
%mapping $f:{\cal S} \rightarrow 2^{\cal S}$, where $f(s) \subseteq {\cal S}$
%determines the set of {\em legal} states for any initial state $s \in {\cal
%S}$. A particle system {\em solves} problem $P=({\cal S}',f)$ if for any
%initial system state $s \in {\cal S'}$, all computations of the particle
%system eventually reach a system state in $f(s)$, and whenever such a system
%state is reached for the first time, the system stays there. A particle system
%{\em decides} problem $P=({\cal S}',f)$ if in addition to solving it, it also
%halts, i.e., all particles are in a final state in $f(s)$. A particle in a
%final state cannot perform any further actions.

\subsection{Universal Coating Problem}
\label{sec:problem}

For any two nodes $v,w \in \Veqt$, the {\em distance} $d(v,w)$ between $v$ and
$w$ is the length of the shortest path in $\Geqt$ from $v$ to $w$. The
distance $d(v,U)$ between a $v \in \Veqt$ and $U \subseteq \Veqt$ is defined
as $\min_{w \in U} d(v,w)$.

In the {\em universal coating problem} we are given an instance $(P,O)$ where
$P$ represents the particle system and $O$ the fixed object to be coated.
Let $V(P)$ be the set of nodes occupied by $P$ and $V(O)$ be the set of
nodes occupied by $O$ (when clear from the context, we may omit the $V(\cdot)$ notation). We call the set of nodes in $\Geqt$ neighboring $O$ the {\em surface (coating) layer}.
%AR09/18: had to change name form boundary to surface layer, since boundary particles are used in the proof with a different meaning
 Let $n$ be the number of particles and $B$ be the
number of nodes in the surface layer.
%$neighboring $V(O)$ in $\Geqt$. We may also call $B$ the size
%of the {\em boundary} of $O$.
An instance is called {\em valid} if the
following properties hold:
\begin{enumerate}
\item The particles are all contracted and start in an \emph{idle} state.

\item The subgraphs of $\Geqt$ induced by $V(O)$ and $V(P) \cup V(O)$, respectively,
are connected, i.e., we are dealing with a single object and the particle
system is connected to the object.

\item The subgraph of $\Geqt$ induced by $\Veqt \setminus V(O)$ is
connected, i.e., the object $O$ does not contain any holes.\footnote{If $O$ does
contain holes, we consider the subset of particles in each connected region of
$\Veqt \setminus V(O)$ separately.}

\item \label{item:noHoleProp} $\Veqt \setminus V(O)$ is $2(\lceil\frac{n}{B}\rceil +1)$-connected. In other
words, $O$ cannot form {\em tunnels} of width less than
$2(\lceil\frac{n}{B}\rceil +1)$.
%Note that a
%width of at least $2 \lceil\frac{n}{B}\rceil$ is needed to guarantee that the
%object can be evenly coated.
\end{enumerate}

%Intuitively, Property~\ref{item:noHoleProp} restricts the allowed concavity of $O$.
Note that a
width of at least $2 \lceil\frac{n}{B}\rceil$ is needed to guarantee that the object can be evenly coated.
See Figure~\ref{fig:tunnel} for an example of an object with a tunnel of width $1$.
%We note that it is of course possible to (unevenly) coat an object that does not fulfill Property~\ref{item:noHoleProp}. However, s
Since coating narrow tunnels requires specific technical mechanisms that complicate the protocol and do not add much to the basic idea of coating, we decided to ignore narrow tunnels completely in favor of a clean presentation.

\begin{figure}
  \centering
  \includegraphics{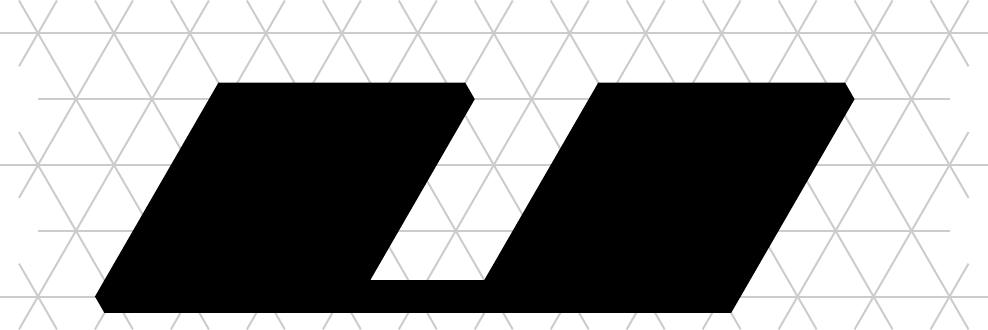}
  \caption{
    An example of an object with a tunnel of width 1.
  }
  \label{fig:tunnel}
\end{figure}

A configuration $C$ is {\em legal} if and only if all particles are contracted
and
\[
  \min_{v \in \Veqt \setminus (V(P) \cup V(O))} d(v,V(O)) \ge \max_{v \in V(P)}
  d(v,V(O))
\]
i.e., the particles are as close to the object as possible, which means that
they {\em coat $O$ as evenly as possible}.

%ZD: the definition is updated:
An algorithm \emph{solves} the universal coating problem if, starting from any valid configuration, it  reaches a {\em stable legal configuration} $C$ in a finite  number of rounds.  A configuration $C$ is said to be stable if no
particle in $C$ ever performs a state change or movement.

\subsection{Our Contributions}
\label{sec:contr}
Our main contribution in this paper is a {\em worst-case work-optimal} algorithm
for the universal coating problem on self-organizing particle systems.
%, as defined in Section~\ref{sec:problem}.
Our {\em Universal Coating Algorithm}
seamlessly adapts to any valid object $O$,
%satisfying the constraints of a valid instance,
uniformly coating the object by forming multiple coating layers if
necessary. As stated in Section~\ref{sec:model}, our particles are anonymous,
do not have any global information (including on the number of particles $n$),
have constant-size memory, and utilize only local interactions.

Our algorithm builds upon many primitives, some of which may be of interest on
their own: The {\em spanning forest} primitive organizes the particles into a
spanning forest which is used to guide the movement of particles while preserving connectivity in the system; the {\em
complaint-based coating} primitive allows the first layer to form, only
expanding the coating of the first layer as long as there is still room and
there are particles still not touching the object; the {\em general layering}
primitive allows the layer $\ell$ to form only after layer $\ell-1$ has been
completed, for $\ell\geq 2$; and a {\em
node-based leader election} primitive elects a position (in $\Geqt$) to house a leader particle, which is used to jumpstart the general layering process. One of the main
contributions of our work is to show how these primitives can be integrated
in a seamless way, with no underlying synchronization mechanisms.

\subsection{Related work}

Many approaches have already been proposed that can potentially be used for
smart coating. One can distinguish between active and passive systems. In
passive systems the particles either do not have any intelligence at all (but
just move and bond based on their structural properties or due to chemical
interactions with the environment), or they have limited computational
capabilities but cannot control their movements. Examples of research on
\emph{passive systems} are DNA self-assembly systems (see, e.g., the surveys
in~\cite{doty2012,patitz2014,Woods2013intrinsic}), population protocols
\cite{AAD+06}, and slime molds~\cite{BMV12,LTT+10}. We will not describe these
models in detail since we are focusing on active systems. In \emph{active
systems}, computational particles can control the way they act and move in
order to solve a specific task. Robotic swarms, and modular robotic systems
are some examples of active programmable matter systems.

Especially in the area of \textit{swarm robotics} the problem of coating
objects has been studied extensively. In swarm robotics, it is usually assumed
that there is a collection of autonomous robots that have limited sensing,
often including vision, and communication ranges, and that can freely move in
a given area. However, coating of objects is commonly not studied as a
stand-alone problem, but is part of \emph{collective transport}
(e.g.,~\cite{WilsonPKBPB14}) or \emph{collective perception} (see respective
section of~\cite{BrambillaFBD13,navarro2012introduction} for a summary of
results). In collective transport a group of robots has to cooperate in order
to transport an object. In general, the object is heavy and cannot be moved by
a single robot, making cooperation necessary. In collective perception, a
group of robots with a local perception each (i.e., only a local knowledge of
the environment), aims at joining multiple instances of individual perceptions
to one big global picture (e.g. to collectively construct a sort of map). Some
research focuses on coating objects as an independent task under the name of
\emph{target surrounding} or \emph{boundary coverage}. The techniques used in
this context include stochastic robot behaviors~\cite{KumarB14,pavlic2013enzyme},
rule-based control mechanisms~\cite{BlazovicsCFC12} and potential field-based
approaches~\cite{BlazovicsLF12}.
%The work in this area follows a variety of goals, for example graph exploration
%(e.g.,~\cite{fl13}), gathering problems (e.g., \cite{AG3,ci12}), shape
%formation problems (e.g.,~\cite{fl08,kilobots}), and to understand the global
%effects of local behavior in natural swarms like social insects, birds, or
%fish (see e.g.,~\cite{BhattacharyyaBCN13,Cha09}).
Surveys of recent results in swarm robotics can be found
in~\cite{Ker12,McL08,BrambillaFBD13,navarro2012introduction}; other samples of
representative work can be found in
e.g.,~\cite{AR10,CP08,DFSY10,DS08,HABFM02}. While the analytic techniques
developed in the area of swarm robotics and natural swarms are of some
relevance for this work, the individual units in those systems have more
powerful communication and processing capabilities than the systems we
consider, and they can move freely.

In a recent paper~\cite{MichailS14}, Michail and Spirakis propose a model
for network construction that is inspired by population
protocols~\cite{AAD+06}. The population protocol model relates to self-organizing particles systems, but is also intrinsically different: agents (which would correspond to our particles) freely move in space and can establish connections to any other agent in the system at any point in time, following the respective probabilistic distribution. In the paper the authors focus on network
construction for specific topologies (e.g., spanning line, spanning star,
etc.). However, in principle, it would be possible to adapt their approach also for
studying coating problems under the population protocol model.

\subsection{Structure of the paper}

Section~\ref{sec:algo} describes our Universal Coating algorithm. Formal correctness and  worst-case work analyses of the algorithm follow in Section~\ref{sec:analysis}. We address some applications of our universal coating algorithm in Section~\ref{sec:applications}, and present our concluding remarks in
Section~\ref{sec:concl}.

\section{Universal Coating Algorithm}
\label{sec:algo}
In this section we present our Universal Coating algorithm:
In Section~\ref{subsec:Pre}, we introduce some preliminary notions;
Section~\ref{subsec:CoatingPrimitives} introduces the  algorithmic primitives used for the coating algorithm; and lastly
Section~\ref{subsec:leaderElection} focuses on the leader election process that is needed in certain instances of the problem.

\subsection{Preliminaries}
\label{subsec:Pre}

%{\bf AR: I would like to change guide-->root and root(used in the analysis section only)-->super-root or supra-root (or something like that). The main reason is consistency with our previous papers, and also since I would rather choose what seems best for the algorithm, and maybe find alternatives for the analysis, since most people will probably only read the algorithms part of our paper (and skip the analysis). Robert, I talked to Christian and Thim about this, when we were discussing whether to change the algorithm terminology in the NatComp submission to be consistent with what we have here, or vice-versa, and we all agreed that we should have for NatComp to be the same as what we had for DNA21, which leaves the terminology used in the spanning forest algorithm here as the exception... This was during the meeting you could not attend. Please let us know what you think}

%\textbf{ZD: I applied the changes root--> super-root and guide --> root. Please double check to make sure I made it everywhere correctly. I put the comment ZD:phraseChange all the places I have changed root to be super-root and I made the original word as comment too to make comparison easier if needed. I changed guide to root with no comment}

We define the set of \emph{states} that a particle can be in as
\emph{idle}, \emph{follower}, \emph{root}, and  \emph{retired}.
In addition to its state, a particle may maintain a constant number of \emph{flags}
(constant size pieces of information to be read by neighboring particles).
While particles are anonymous, when a particle $p$ sets a flag of type $x$ in its shared memory, we will denote it by $p.x$ (e.g., $p.parent$, $p.dir$, etc.), so that ownership of the respective flag becomes clear.
In our proposed algorithm, we assume that every time a particle contracts, it contracts out of its tail.
Therefore, a node occupied by the head of a particle $p$ still is occupied by $p$ after a contraction.

We define a \emph{layer} as the set of nodes $v$ in $\Geqt$ that are equidistant to the object $O$.
More specifically a node $v$ is in layer $\ell$ if $d(v, V(O)) = \ell$; in particular the surface coating layer defined earlier corresponds to layer 1.
Any root or retired particle $p$ stores a flag $p.layer$
indicating the layer number of the node occupied by the head of $p$.
We say a layer is \emph{filled} or \emph{complete} if all nodes in that layer are occupied with retired particles.
In order to respect the particles' constant-size memory constraints, we take all layer numbers modulo $4$.
However, for ease of presentation,
%When clear from the context,
we will omit the modulo 4 computations in the text, except for in the pseudocode description of the algorithms.

Each root particle $p$ has a flag storing a port label $p.down$
pointing to an occupied node adjacent to its head in layer $p.layer - 1$ or in the object if $p-layer=1$.
%\textbf{RG: What about roots on layer 1? They have to point to a position occupied by the object, right?}
Moreover, $p$ has two additional flags, $p.CW$ and $p.CCW$, which are also port labels.
Intuitively, if $p$ continuously moves by expanding in direction $p.CW$ (resp., $p.CCW$) and then contracting,
it moves along a \emph{clockwise} (resp. \emph{counter-clockwise}) path around the connected structure
consisting of the object and retired particles.
Formally, $p.CW$ is the label of the first port to a node $v$ in \emph{counter-clockwise (CCW) order} from $p.down$
such that either $v$ is
occupied by a particle $q$
%(that must be a guide or retired)
with $q.layer = p.layer$,
or $v$ is unoccupied (in the latter, $v$ may be a node on layer $p.layer$ or $p.layer -1$).
%not part of the object and not occupied by a guide or retired particle.
%In the latter, $v$ may be a node on layer $p.layer$ or $p.layer -1$, in which case $p$ updates $p.layer$ accordingly.
We define $p.CCW$ analogously, following a \emph{clockwise (CW) order} from $p.down$. Figure~\ref{fig:coatingLayers} illustrates the different layers around an object, and also  CW and CCW traversals of those.

\begin{figure}
   \centering
   \includegraphics[width = 0.6\textwidth]{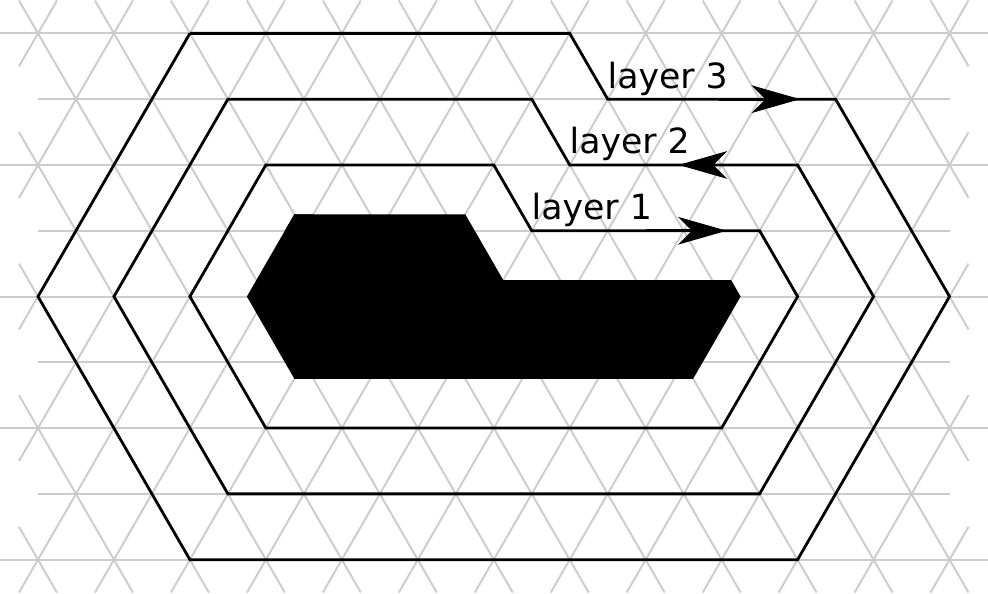}
     \caption{
     We illustrate the first three coating layers with respect to the given object (represented by the nodes in $\Geqt$ shaded in black); we also illustrate the direction in which these layers will be filled by our algorithm --- CW for odd layers, and CCW for even layers ---  as determined in Section~\ref{subsec:CoatingPrimitives}.
     %The direction in which layers are filled (and in which particles retire within the layer) alternates between CW and CCW.
   }
   \label{fig:coatingLayers}
\end{figure}

\subsection{Coating Primitives}
\label{subsec:CoatingPrimitives}
Our algorithm can be decomposed into a set of primitives, which are all concurrently executed by the particles, as we briefly described in Section~\ref{sec:contr}.
Namely the algorithm relies on the following key primitives:
the \emph{spanning forest primitive},
the \emph{complaint-based coating primitive} used to establish the first layer of coating,
the \emph{general layering primitive},
and a \emph{node-based} (rather than particle-based) \emph{leader election} primitive
that works even as particles move, and that is used to jumpstart the general layering primitive.
One of the main contributions of our work is to show how these primitives can be put to work together
in a seamless way and with no underlying synchronization mechanisms.\footnote{A video illustrating a fully asynchronous execution of our universal coating algorithm can be found in~\cite{sops-webpage}.}

The \textbf{\emph{spanning forest primitive}} (Algorithm~\ref{alg:spanningForestAlgorithm})
organizes the particles in a spanning forest, in which the roots of the trees will be in state {\em root} and will determine the direction of movement which is specified by a port label $p.dir$;
the remaining non-retired particles follow the root particles %(guide particles)
using handovers. The main benefit of organizing the particles in a spanning forest connected to the surface is that it provides a relatively easy mechanism for particles to move, following the tree paths, while maintaining connectivity in the system (see~\cite{DNA, NANOCOM} for more details). %\textbf{RG: Broken references.}
All particles are initially \emph{idle}.
A particle $p$ becomes a \emph{follower} when it sets a flag $p.parent$ corresponding to the port
leading to its parent in the spanning forest
(any adjacent particle $q$ to $p$ can then easily check if $q$ is a child of $p$).
As the root particles find final positions according to the partial coating of the object,
they stop moving and become retired.
Namely, a root particle $p$ becomes \emph{retired}
when it encounters another retired particle across the direction $p.dir$.

\begin{algorithm}
  A particle $p$ a acts depending on its state as described below: \\
  \begin{tabularx}{\textwidth}{lX}
    \textbf{idle}: &
    If $p$ is connected to the object $O$, it becomes a \emph{root} particle, makes the current node it occupies a \emph{leader candidate position}, and starts running the leader election algorithm described in Section~\ref{subsec:leaderElection}.
    If $p$ is connected to a \emph{retired} particle, $p$ also becomes a \emph{root} particle.
    If an adjacent particle $p'$ is a root or a follower,
    $p$ sets the flag $p.parent$ to the label of the port to $p'$,
    puts a \emph{complaint flag} in its local memory, and becomes a \emph{follower}.
    If none of the above applies, $p$ remains idle.
    \\ \\

    \textbf{follower}: &
    If $p$ is contracted and connected to a retired particle or to $O$, then $p$ becomes a \emph{root} particle.
    Otherwise, if $p$ is expanded, it considers the following two cases:
    %AR: I removed this case since it will be covered from the point-of-view of p'
    %$(i)$ if $p$ is contracted and $p$'s parent $p'$  is expanded, then $p$ initiates (and takes part in) {\sc Handover$(p')$};
    %%expands in the direction given by $p.parent$ in a handover with $p'$, and may need to adjust $p.parent$ to still point to particle $p'$ after the handover;
    $(i)$ if $p$ has a contracted child particle $q$, then $p$ initiates {\sc Handover$(p)$};
    $(ii)$ if $p$ has no children and no idle neighbor, then $p$ contracts.
    Finally, if $p$ is contracted, it runs the function {\sc ForwardComplaint$(p,p.parent)$} described in Algorithm~\ref{alg:complaint}.
    \\ \\

    \textbf{root}: &
    If particle $p$ is on the surface coating layer, $p$ participates in the leader election process described in Section~\ref{subsec:leaderElection}.    If $p$ is contracted, it first executes  {\sc MarkerRetiredConditions$(p)$} (Algorithm~\ref{alg:retiredCondition}), and becomes {\em retired}, and possibly also a {\em  marker}, accordingly; if $p$ does not become retired, it calls {\sc LayerExtension $(p)$} (Algorithm~\ref{alg:boundaryDirectionAlgorithm}). If $p$ is expanded, it considers the following two cases:
    $(i)$ if $p$ has a contracted child, then
    $p$ initiates {\sc Handover$(p)$};
        $(ii)$ if $p$ has no children
    and no idle neighbor, then $p$ contracts.
    Finally, if $p$ is contracted, it runs {\sc ForwardComplaint$(p,p.dir)$} (Algorithm~\ref{alg:complaint}).
    \\ \\

    \textbf{retired}: &
    $p$ clears a potential complaint flag from its memory and performs no further action.
    \\
  \end{tabularx}

  \caption{Spanning Forest Primitive}
  \label{alg:spanningForestAlgorithm}
\end{algorithm}

%AR12/10: changed the paragraph below a bit
%complaint-based algorithm
Recall that $B$ denotes the number of nodes on the surface coating layer (layer 1).
%In case $B$ is larger than the number of particles in the system $n$,
%\textbf{RG: Missing reference. Also I do not understand the part about "variant surface scenarios".}
We need to ensure that once $\min\{n,B\}$ particles are on layer~$1$, they stop moving and the coating is complete, independent of how $B$ compares to $n$ (i.e., whether $n\leq B$ or not); in addition, we would like to efficiently coat one more surface scenario, namely that of coating just a bounded segment of the surface, as we explain in Section~\ref{sec:applications}.
In order to be able to seamlessly adapt to all possible coating configurations, we use our novel \textbf{\emph{complaint-based coating primitive}} for the first layer,
which basically translates into having the root particles (touching the object) open up one more position on layer~$1$  only if there exists a follower particle that remains in the system.
This is accomplished by having each particle that becomes a follower
%, hence not initially touching the object,
generate a \emph{complaint flag},
which will be forwarded by particles in a pipeline fashion from children to parents through the spanning forest
and then from a root $q$ to another root at $q.dir$,
until it arrives at a root particle $p$ with an unoccupied neighboring node at $p.dir$ (we call such a particle $p$ a {\em super-root}).
Upon receiving a complaint flag, a super-root $p$ consumes the flag and expands into the unoccupied node at $p.dir$.
The expansion will eventually be followed by a contraction of $p$,
which will induce a series of expansions and contractions of the particles
on the path from $p$ to a follower particle $z$, eventually freeing a position on the surface coating layer to be taken by $z$.
In order to ensure that the consumption of a complaint flag will indeed
result in one more follower touching the object,
one must give higher priority to a follower child particle in a handover operation,
as we do in Algorithm~\ref{alg:handover}.
The complaint-based coating phase of the algorithm will terminate either once all complaint flags are consumed
or when layer~$1$ is filled with contracted particles.
In either case, the particles on layer~$1$ will move no further. Figure~\ref{fig:tokens} illustrates the complaint-based coating primitive.

\begin{figure}
  \centering
  \includegraphics{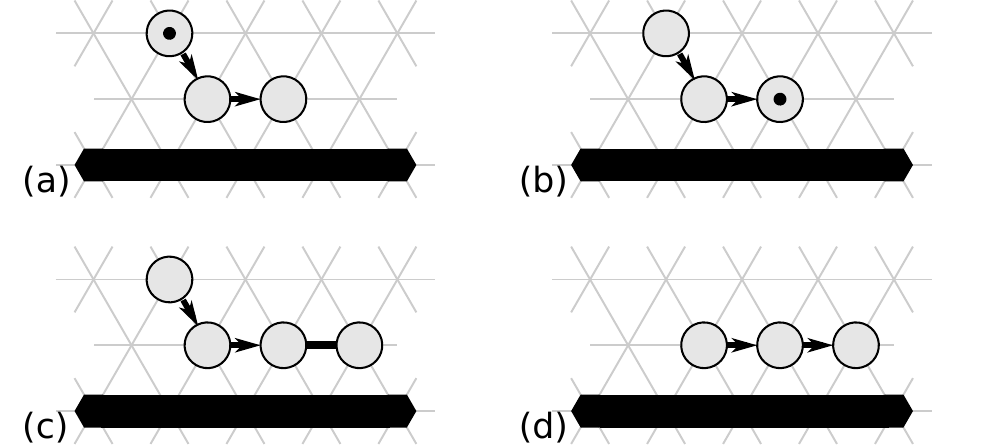}
  \caption{
    Complaint-based coating primitive:
    Particles are shown as grey circles.
    In (a), a follower particle generates a complaint flag (depicted as a black dot within the particle) that is then forwarded to a super-root (b) causing the super-root to expand into an unoccupied node (c). After a series of handovers, the follower particle that generated the complaint flag can move to a position on the surface (d).
  }
  \label{fig:tokens}
\end{figure}

\begin{algorithm}
  \begin{algorithmic}[1]
    %\Function{Handover}{$p$}
    %\Comment {$p$ is expanded}
    %\State Let $q$ be a follower child of $p$; if $p$ has no follower child, then let $q$ be a boundary such that edge $(p,q)$ is flagged $q.dir$.
    % \Statex This handover operation will always give priority to a follower child if layer $1$ is still not complete, so that the complaint-based coating for layer $1$ works correctly
    \If {$p.layer = 1$ and $p$ has a follower child $q$}
      \If {$q$ is contracted}
        \State $p$ initiates a handover with particle $q$
      \EndIf
    \Else
      \If {$p$ has any contracted (follower or root) child $q$}
        \State $p$ initiates a handover with particle $q$
      \EndIf
    \EndIf
    %\EndFunction
  \end{algorithmic}
  %}
  \caption{{\sc Handover} $(p)$}
  \label{alg:handover}
\end{algorithm}

\begin{algorithm}
  \begin{algorithmic}[1]
    \If {$p$ holds a complaint flag \textbf{and} $p$'s parent does not hold a complaint flag}
      \State $p$ forwards the complaint flag to the particle given by $p.parent$
    \EndIf
  \end{algorithmic}
  \caption{{\sc ForwardComplaint$(p,i)$}}
  \label{alg:complaint}
\end{algorithm}

Once layer~$1$ is complete and if there are still follower particles in the system,
the \textbf{\emph{general layering primitive}} steps in, which will build further coating layers.
We accomplish this by electing a \emph{leader marker particle} on layer~$1$
(via the \textbf{\emph{leader election primitive}} proposed in Section~\ref{subsec:leaderElection}).
This leader marker particle will be used to determine a ``beginning" (and an ``end") for layer~$1$
and allow the particles on that layer to start retiring according to the retired condition
given in Algorithm~\ref{alg:retiredCondition}
(the leader marker particle will be the first retired particle in the system).
Once a layer~$\ell$ becomes completely filled with retired (contracted) particles,
a new marker particle will emerge on layer $\ell + 1$, and start the process of building this layer (i.e., start the process of retiring particles on that layer) according to Algorithm~\ref{alg:retiredCondition}.
A marker particle on layer $\ell + 1$ only emerges
if a root particle $p$ connects to the marker particle $q$ on layer $\ell$ via its marker port
and if $q$ verified locally
that layer $\ell$ is completely filled (by checking whether $q.CW$ and $q.CCW$ are both retired).

\begin{figure}
  \centering
  \includegraphics{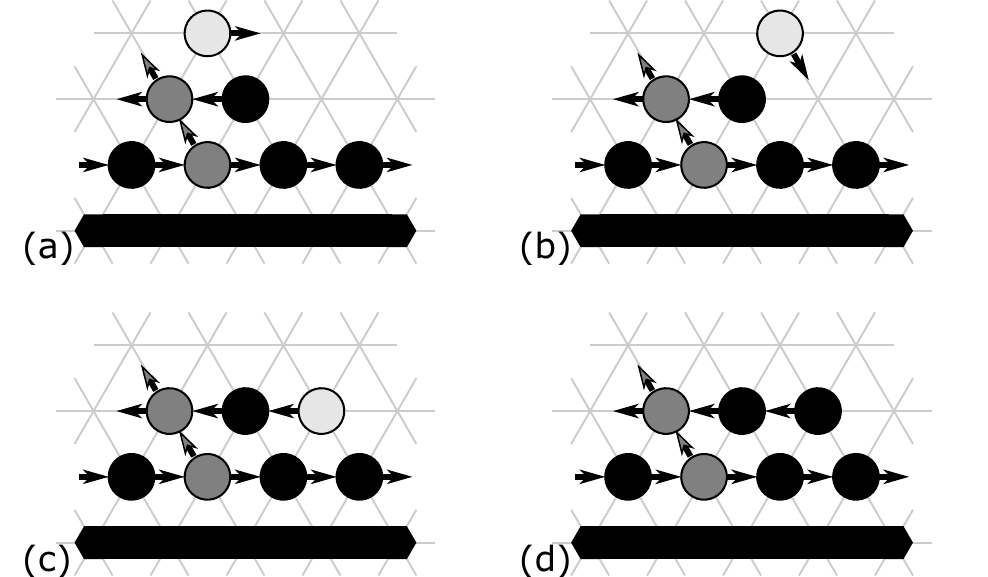} 
  \caption{
	General layering primitive: Retired particles are shown as black circles, other than (retired) marker particles which are shown in dark grey (the dark grey arrows represent the  marker edges); a root particle is depicted in light grey. Black arrows show the current direction of movement  (given by the {\em dir} flag) for each particle (which becomes irrelevant once a particle retires). 
(a) The root particle $p$ is located on layer $\ell =3$;  (b) particle $p$ moves in $CW$ direction over retired particles on layer $\ell-1$; (c) after a series of expansions and contractions following $p.dir$, $p$ arrives at  an unoccupied neighboring node on layer $\ell-1$; (d) since $p.dir$ leads to a retired particle, $p$ retires too. 
     }
  \label{fig:layeringWithMarker}
\end{figure}

With the help of the marker particles --- which can only be established after layer~$1$ was completely filled (and hence, we must have $B\leq n$) ---
we can replace the complaint-based coating algorithm of layer~$1$ with a simpler coating algorithm
for the higher layers, where each root particle $p$ just moves in $CW$ or $CCW$ direction
(depending on its layer number) until $p$ encounters a retired particle on the respective layer and retires itself.
%In order to coat the surface of the object we need the boundarys of the spanning forest to move in the correct direction (while the %remaining particles follow behind) for the respective layer being formed.
% resulting in the system flattening out towards the layer being formed.
More precisely, each contracted root particle $p$ on layer $\ell$ tries to extend this layer
by expanding into an unoccupied position on layer $\ell$,
or by moving into an unoccupied position in layer $\ell - 1$
(when $p.layer$ will change to $\ell - 1$ accordingly),
following the direction of movement given by $p.dir$. Figure~\ref{fig:layeringWithMarker} illustrates this process. 
%{\bf AR: A figure showing how a particle may fill the next position in layer l-1 or the next position in layer l could be helpful here.}
The direction $p.dir$ is set to $p.CW$ (resp., $p.CCW$) when $p.layer$ is odd (resp., even), as illustrated in Figure~\ref{fig:coatingLayers}.
%, which is $p.CW$ if $p$'s layer number is even and  $p.CCW$ otherwise.
Alternating between $CCW$ and $CW$ movements for the particles in consecutive layers ensures
that a layer $\ell$ is completely filled with retired particles
before particles start retiring in layer $\ell + 1$,
which is crucial for the correctness of our layering algorithm.
%{\bf AR; could probably say more here, or make things clearer}

\begin{algorithm}
  %{\fontsize{8pt}{9}
  \begin{algorithmic}[1]
    %   \Function {LayerExtension}{$p$}
    % \Statex
    \Statex{\bf Calculating $p.layer$, $p.down$ and $p.dir$}
    \State The layer number of any node occupied by the object is equal to 0.%AR09/17I changed the layer number back to 0, to conform with the changes in layer numbers elsewhere in the paper \textbf{ZD: I changed these parts a bit}
    \State Let $q$ be any neighbor of $p$ with smallest layer number (modulo $4$).
    \State $p.down \; \gets \; \mbox{$p$'s label for port leading to } q$
    \State $p.layer=(q.layer+1)\mod \: 4 $
    %\State $p.layer \; \gets \; (l+1)\mod \: 4$ %\Comment{Calculating Layer Number}
    %       \Statex
    %       \Statex \textbf{Phase 2: Calculating $p.dir$}
    %\State Let $i$ be the edge of $p$ with flag $p.down$
    \State {\sc clockwise} $(p, p.down)$ \Comment{Computes CW \& CCW directions}
    %       \Statex \Comment{adjacent to $p$ in $CW$ and $CCW$ order}
    \If {$p.layer$ is {\em odd}}
      %\State  $p.CW= ${\sc clockwise} $(p, +1)$
      \State $p.dir  \; \gets \; p.CW$
    \Else
    %\State  $p.CCW= ${\sc clockwise} $(p, -1)$
      \State $p.dir \; \gets \; p.CCW$
    \EndIf

    \Statex
    \Statex \textbf{Extending layer $p.layer$}
    \If { the position at $p.dir$ is unoccupied, and either $p$ is not on the first layer,
    %$p.layer>1$ \textbf{ZD: I would still say "$p$ is not connected to the surface", instead since having mod there is only one layer with $p.layer=0$ which we want to exclude while we have several layers with $p.layer=0$ to not exclude}
    or $p$ holds a complaint flag}
      \State $p$ expands in direction $p.dir$
      \State $p$ consumes its complaint flag, if it holds one
    \EndIf
    %   \EndFunction
  \end{algorithmic}

  \caption{{\sc LayerExtension $(p)$}}
  \label{alg:boundaryDirectionAlgorithm}
\end{algorithm}

\begin{algorithm}
  \begin{algorithmic}[1]
    % \Function{clockwise}{$p, i$}
    \State $j \; \gets \;  i$,  $k \; \gets \; i$
    \While{edge $j$ is connected to the object or to a retired particle with layer number $p.layer-1$}%$(p.layer -1) \mod 4$}
      \State $j\; \gets \; (j - 1) \mod 6$
    \EndWhile
    \State $p.CW \; \gets \; j$
    %\State $j \; \gets \; i$
    \While{edge $k$ is connected to the object or to a retired particle with layer number $p.layer-1$}%$(p.layer -1) \mod 4$}
      \State $k \; \gets \; (k + 1) \mod 6$
    \EndWhile
    \State $p.CCW \; \gets \; k$
    %\State \textbf{return} $i$
    %\EndFunction
  \end{algorithmic}
  %}

  \caption{{\sc Clockwise $(p,i)$}}
  \label{alg:clockwise}
\end{algorithm}

\begin{algorithm}
  %{\fontsize{9pt}{9}
  \begin{algorithmic}[1]
    %   \State For simplicity, we treat the end of surface node in an open bounded object as a {\em SingleMarker} particle $q$ with boundary direction flagged by $q.marker$.
    %   \vspace{.05in}
    \Statex{\bf First marker condition:}
    \If {$p$ is {\em leader particle} }
      \State $p$ becomes a {\em retired}  particle
      \State $p$ sets the flag $p.marker$ to be the label of a port leading to a node guaranteed not to be on layer $p.layer$ --- e.g., by taking the average
      direction of $p$'s two neighbors in layer $1$
      (by now complete)
      %\State $p$ sets the flag $p.Marker$ of its edge adjacent to a node of layer $p.layer+1$ on $\Geqt$
    \EndIf

    \Statex
    \Statex{\bf Extending Layer Markers:}

    \If {$p$ is connected to a marker $q$ and the port $q.marker$ points towards $p$}
    %         \If {$q$ is a {\em SingleMarker} }
    %             \State  $p$ becomes a {\em SingleMarker} and {\em retired} particle
    %               \State $p$ sets the flag $p.marker$ of its edge
    %               \Statex \mbox{\ \ \ \ \ \ \ \ } opposite to edge flagged $q.marker$
    %           \Else \Comment{$p$ is DoubleMarker}
      \If {both $q.CW$ and $q.CCW$ are retired} %\Comment{LayerFilled Condition}
        \State $p$ becomes a {\em retired}  particle
        \State $p$ sets the flag $p.marker$ to the label of the port opposite of the port connecting $p$ to $q$
      \EndIf
      %           \EndIf
    \EndIf

    \Statex
    \Statex{\bf Retired Condition:}
    %\State  Let $i$ be the edge with flag $p.dir$ returned by Algorithm~\ref{alg:boundaryDirectionAlgorithm}.
    \If{edge $p.dir$ is occupied by a retired particle}
      \State $p$ becomes {\em retired}
    \EndIf
    %\If {$p$ is DoubleMarker or SingleMarker}
    %\State $p$ sets the flag $p.marker$ of its edge opposite to $q.marker$ ......ta inja. oon average o ina ra ham be jaye khodesh biar...
    %\EndIf
  \end{algorithmic}

  \caption{{\sc MarkerRetiredConditions$(p)$}}
  \label{alg:retiredCondition}
\end{algorithm}

\subsection{Leader Election Primitive}
\label{subsec:leaderElection}
In this section, we describe the process used for electing a leader among the particles that touch the object. Note that only particles in layer~$1$ will ever participate in the leader election process. A leader will only emerge if $B \le n$; otherwise the process will stop at some point without a leader being elected.
%If the surface is closed, there are no set nodes that determine the endpoints of the surface, and which can be used to verify %if a layer has been completed. If we have a closed surface, we need to verify if layer $iHowever, if the closed object's surface %is smaller in length than the number of particles in the system, the particles will need to form multiple layers over the surface.
As discussed earlier, a leader is elected on layer~$1$ to provide a ``checkpoint" (a {\em marker} particle) that the particles can use in order to determine whether the layer has been completely filled (and a leader is only elected after this happens).
%In that case particles on layer~$1$ start retiring and higher layers can be formed.
%Basically, the leader particle doubles as both a CW and a CCW layer 0 marker particle in the bounded open surface case, %as we saw in Algorithm~\ref{alg:retiredCondition}.

%The leader (marker) particle will create other marker particles in the layers above as needed in order for the other layers to %be formed.

The leader election algorithm we use in this paper is a slightly modified version of the leader election algorithm presented in~\cite{DNA} that can tolerate particles moving around on the surface layer while the leader election process is progressing (in~\cite{DNA}, leader election runs on a system of static particles). Hence, for the purpose of universal coating, we will abstract the leader election algorithm to conceptually run on the \emph{nodes} in layer~$1$, and not on the particular particles that may occupy these nodes at different points in time. The particles on layer 1 will simply provide the means for running the leader election process on the respective positions, storing and transferring  all the flags (which can be used to implement the tokens described in~\cite{DNA}) that are needed for the leader competition and verification. An expanded particle $p$ on layer 1, whose tail occupies node $v$ in layer~$1$,
that is about to perform a handover with contracted particle $q$ will pass all the information associated with $v$ to $q$ using the particles'  local shared memories.
If a particle $p$ occupying position $v$ would like to forward some leader election information to a node $w$ adjacent to $v$ that is currently unoccupied, it will wait until either $p$ itself expands into $w$, or another particle occupies node $w$. It is important to note that according to the complaint-based coating algorithm that we run on layer~$1$, if a node $v$ in layer~$1$ is occupied at some time $t$, then $v$ will never be left unoccupied after time $t$.

Here we outline the differences between the leader election process used in this paper and that of~\cite{DNA}:

\begin{list}{\labelitemi}{\setlength{\itemsep}{0in}\setlength{\leftmargin}{.2in}}
  \item Only the nodes on layer~$1$ that initially hold particles start as {\em leader node candidates}. Other nodes in layer~$1$ will take part in the leader node election process by forwarding any tokens between two consecutive leader node candidates, as determined for the leader election process for a set of static particles forming a cycle in~\cite{DNA}. Note that layer 1 is a cycle on $\Geqt$.
  \item The leader election process will determine which leader node candidate in layer~$1$ will emerge as the unique {\em leader node}. The {\em leader particle} is then chosen as described below.
  \item If particle $p$ is expanded, it will hold the flags and any other information necessary for the leader election process corresponding to each node $p$ occupies (head and tail nodes) independently. In other words, an expanded particle emulates the leader election process for two nodes on the surface layer.
  \item A particle $p$ occupying node $v$ forwards a flag $\tau$ to the node $w$ in CW (or CCW) direction along the surface layer only if node $w$ is occupied by a particle $q$  (note that $q$ may be equal to $p$, if $p$ is expanded) and $q$ has enough space in its (constant-size) memory associated with node $w$; otherwise $p$ continues to hold the flag $\tau$ in its shared memory associated with node $v$.
  \item If $p$ is expanded along an edge $(v,w)$ and wants to contract into node $w$, there must exist a particle $q$ expanding into $v$ (due to the complaint-based mechanism), and hence $p$ will transfer all of its flags currently associated with node $v$ to particle $q$.
\end{list}

%AR09/18: please do NOT put CW and CCW in math mode (within $ signs) since they are just abbreviations!!
After the solitude verification phase in the leader election algorithm of~\cite{DNA} is complete, there will be just one leader node $v$ left in the system. Once $v$ is elected a leader node, a contracted particle $p$ occupying this position will check if layer $1$ is completely filled with contracted particles. To do so, when a contracted particle $p$ occupies node $v$ it will generate a single {\em CHK} flag which it will forward to its CCW neighbor $q$ {\em only if $q$ is contracted}. Any particle $q$ receiving a {\em CHK} flag will also only forward the flag to its CCW neighbor $z$ if and only if $z$ is contracted. If the {\em CHK} flag at a particle $q$ ever encounters an expanded CCW neighbor, the flag is held back until the neighbor contracts.  Additionally, the particle at position $v$ sends out a {\em CLR} flag to its CW neighbor as soon as it expands.
This flag is always forwarded in CW direction. If a {\em CLR} and  a {\em CHK} flag meet at some particle, the flags cancel each other out. If at some point in time, a particle $p$ at node $v$ receives a {\em CHK} flag from its CW neighbor in layer $1$, it implies that layer $1$ must be completely filled with contracted particles (and the complaint-based algorithm for layer $1$ has converged), and at that time this contracted particle $p$ elects itself the {\em leader particle}, setting the flag $p.leader$. Note that the leader election process itself does not incur any additional particle expansions or contractions on layer 1, only the complaint-based algorithm does.

\section{Analysis}
\label{sec:analysis}

In this section we show that our algorithm eventually solves the coating
problem, and we bound its worst-case work.

%Recalling from Section~\ref{subsec:CoatingPrimitives} we denote the node in
%$G_{eqt}$ reached from a follower $p$ via the edge labeled with a flag
%$parent$ as $p.parent$.
We say a particle $p'$ is the \emph{parent} of particle $p$ if $p'$
occupies the node in direction $p.parent$. Let an {\em active} particle be a particle in either
follower or root state. We call an active particle a {\em boundary particle} if it
has the object or at least one retired particle in its neighborhood, otherwise it is a \emph{non-boundary particle}.
A boundary particle is either a root or a follower, whereas non-boundary particles are always followers.
Note that throughout the analysis we ignore the modulo computation of layers done by the particles, i.e., layer $1$ is the unique layer of nodes with distance $1$ to the object.

Given a configuration $C$, we define a
directed graph $A(C)$ over all nodes in $\Geqt$ occupied by active particles
in $C$.  For every
expanded active particle in $C$, $A(C)$ contains a directed edge from the tail
to the head node of $p$. For every non-boundary particle $p$, $A(C)$ has a directed edge from the head of $p$ to
$p.parent$, if $p.parent$ is occupied by an active particle, and for every
boundary particle $p$, $p$ has a directed edge from its head to the
node in the direction of $p.dir$ as it would be calculated by
Algorithm~\ref{alg:boundaryDirectionAlgorithm}, if $p.dir$ is occupied by an
active particle. The \emph{ancestors} of a particle $p$ are all nodes
reachable by a path from the head of $p$ in $A(C)$. For each particle $p$ we
denote the ancestor that has no outgoing edge with
%ZD:phraseChange $p.root$
$p.superRoot$, if it exists.
Certainly, since every node has at most one outgoing edge in $A(C)$, the nodes
of $A(C)$ can only form a collection of disjoint trees or a ring of trees. We
define a \emph{ring of trees} to be a connected graph consisting of a single
directed cycle with trees rooted at it.

First, we prove several safety conditions, and then we prove various liveness
conditions that together will allow us to prove that our algorithm solves the
coating problem.

\subsection{Safety}

Suppose that we start with a valid instance $(P,O)$, i.e., all particles in
$P$ are initially contracted and idle and $V(P) \cup V(O)$ forms a single connected
component in $\Geqt$, among other properties. Then the following properties
hold, leading to the fact that $V(P) \cup V(O)$ stays connected at any time.

\begin{lemma}
  \label{lem:NoHole}
  At any time, the set of retired particles forms completely filled layers
  except for possibly the current topmost layer $\ell$,
  which is consecutively filled with retired particles in $CCW$ direction (resp. $CW$ direction)
  if $\ell$ is odd (resp. even).
\end{lemma}

\begin{proof}
  From our algorithm it follows that the first particle that retires is the leader particle, setting its marker flag in a direction not adjacent to a position in layer 1.
  The particles in layer 1 then retire starting from the leader in $CCW$ direction around the object.
  Once all particles in layer 1 are retired, the first particle to occupy the adjacent position to the leader via its marker flag direction will retire and become a marker particle on layer 2, extending its marker flag in the same direction as the flag set by the marker (leader) on layer 1.
  Starting from the marker particle in layer 2, other contracted boundary particles can retire in $CW$ direction along layer 2.
  Once all particles in layer 2 are retired,  the next layer will start forming.
  This process continues inductively, proving the lemma.
\end{proof}

The next lemma characterizes the structure of $A(C)$.

\begin{lemma}
  \label{lem:Trees}
	%ZD:phraseChange: We changed the lemma statement and the proof for Claim 1 a bit
  At any time, $A(C)$ is a forest or a ring of trees.
  Any node that is a %root
	 super-root (i.e., the root
	of a tree in $A(C)$) or part of the cycle in the ring of trees
  is connected to the object or to a retired particle.
\end{lemma}
\begin{proof}
An active particle can either be a follower or a root. First, we
show the following claim.

\begin{claim}
At any time, $A(C)$ restricted to non-boundary particles forms a forest.
\end{claim}
\begin{proof}
%ZD:phraseChange: within this proof:
 Let $A'(C)$  be the induced subgraph of $A(C)$ by the non-boundary particles only.
Certainly, at the very beginning, when all particles are still idle, the claim
is true. So suppose that the claim holds up to time $t$. We will show that it
then also holds at time $t+1$. Suppose that at time $t+1$ an idle particle
$p$ becomes active. If it is a non-boundary particle (i.e., a follower), it sets
$p.parent$ to a node occupied by a particle $q$ that is already active, so it
extends the tree of $q$ by a new leaf, thereby maintaining a tree. Edges can
only change if followers move. However, followers only move by a handover
or a contraction, thus a handover can only cause a %root
 follower and its incoming edges
to disappear from $A'(C)$ (if that follower becomes a boundary particle), and an
 isolated contraction, can only cause a leaf and its outgoing edge to disappear from $A'(C)$, so a
tree is maintained in $A'(C)$ in each of these cases.
\end{proof}

Next we consider $A(C)$ restricted to boundary particles.

\begin{claim}
At any time, $A(C)$ restricted to boundary particles forms a
forest or a ring.
\end{claim}
\begin{proof}
The boundary particles always occupy the nodes adjacent to retired particles
or the object. Therefore, due to
Lemma~\ref{lem:NoHole}, the boundary particles either all lie in a single
layer or in two consecutive layers. Since the layer numbers uniquely specify
the movement direction of the particles, connected boundary
particles within a layer are all connected in the same orientation. Therefore,
if these particles all lie in a single layer, they can only form a directed
list or directed cycle in $A(C)$, proving the claim. If they lie in two
consecutive layers, say, $\ell$ and $\ell-1$, then $\ell-1$ must contain at
least one retired particle, so the nodes occupied by the boundary
particles in layer $\ell-1$ can only form a directed list. If there are at
least two boundary particles in layer $\ell-1$, this must also be
true for the nodes occupied by the boundary particles in layer
$\ell$ because according to Lemma~\ref{lem:NoHole} there must be at least two
consecutive nodes in layer $\ell-1$ not occupied by retired particles.
Moreover, it follows from the algorithm that $p.dir$ of a boundary particle
can only point to the same or the next lower layer of $p$, implying that in
this case $A(C)$ restricted to the nodes occupied by all boundary
particles forms a forest.
\end{proof}

Since a boundary particle $p$ never connects to a non-boundary particle the way
$p.dir$ is defined, and a follower without an outgoing edge in $A(C)$
restricted to the non-boundary particles must have an outgoing edge to a boundary
particle (otherwise it is a boundary particle itself), $A(C)$ is a
forest or a ring of trees. The second statement of the lemma follows from the
fact that every boundary particle must be connected to the object or
a retired particle.
\end{proof}

Finally, we investigate the structure formed by the idle particles.

\begin{lemma}
\label{idleToNonidle}
  At any time, every connected component of idle particles is connected to at least one non-idle particle or the object.
\end{lemma}
\begin{proof}
  Initially, the lemma holds by the definition of a valid instance.
  Suppose that the lemma holds at time $t$ and consider a connected component of idle particles.
  If one of the idle particles in the component is activated, it may either stay idle or change
  to an active particle, but in both cases the lemma holds at time $t+1$.
  If a retired particle that is connected to the component is activated, it does not move.
  If a follower or root particle that is connected to the component is activated,
  that particle cannot contract outside of a handover with another follower or root
  particle, which implies that no node occupied by it is given up by the
  active particles.
  So in any of these cases, the connected component of idle particle remains connected to a non-idle particle.
%  The activation of any other particle obviously cannot disturb the connection
%  between the component and a non-idle particle.
  Therefore, the lemma holds at time $t + 1$.
\end{proof}

The following corollary is consequence of the previous three lemmas.

\begin{corollary}
  At any time, $V(P) \cup V(O)$ forms a single connected component.
\end{corollary}

\begin{lemma}
\label{lem:complaintTokenEquality}
  At any time before the first particle retires,
  in every connected component $G$ of $A(C)$,
  the number of expanded boundary particles in $G$ plus the number of complaint flags in $G$
  is equal to the number of non-boundary particles in $G$.
\end{lemma}
\begin{proof}
  Initially, the lemma holds trivially.
  Suppose the lemma holds at time $t$ and consider the next activation of a particle.
  We only discuss relevant cases.
  If an idle particle becomes a non-boundary particle (i.e., it is not connected to the object but joins a connected component), it also generates a complaint flag.
  So both the number of non-boundary particles
  and the number of complaint flags increases by one for the component the particle joins.
  If a non-boundary particle expands as part of a handover with a boundary particle,
  both the number of expanded boundary particles and the number of non-boundary particles
  decrease by one for the component.
  If a boundary particle expands as part of a handover, that handover must be with another boundary particle,
  so the number of expanded boundary particles remains unchanged for that component.
  Since by our assumption there is no retired particle, all boundary particles are in layer 1.
  Hence, for a boundary particle to expand outside of a handover, it has to consume a complaint flag.
  This increases the number of expanded boundary particles by one and decreases the number of complaint flags by one.
  Finally, an expansion of a boundary particle outside of a handover can connect two components of $A(C)$.
  Since the equation given in the lemma holds for each of these components individually,
  it also holds for the newly built component.
\end{proof}

\subsection{Liveness}

We say that the particle system \emph{makes progress} if (i) an idle particle
becomes active, or (ii) a movement (i.e., an expansion, handover, or
contraction) is executed, or (iii) an active particle retires. In the
following, we always assume that we have a fair activation sequence for the
particles.

Before we show under which circumstances our particle system eventually makes
progress, we first establish some lemmas on how particles behave during the
execution of our algorithm.

\begin{lemma}
\label{lem:idleToActive}
Eventually, every idle particle becomes active.
\end{lemma}
\begin{proof}
As long as an idle particle exists, there is always an idle particle $p$ that
is connected to a non-idle particle or the object according to
Lemma~\ref{idleToNonidle}. The next time $p$ is activated $p$ becomes active
according to Algorithm~\ref{alg:spanningForestAlgorithm}. Therefore,
eventually all particles become active.
\end{proof}

The following statement shows that even though
% ZD:phraseChange
 %roots
 super-roots
can be followers, they
will become a boundary particle the next time they are activated.

% ZD:phraseChange: the lemma and proof has changed slightly
\begin{lemma}
\label{lem:followerRootToBoundary} In every tree of $A(C)$,
 every boundary particle in
%a root in
 the follower state
enters a root state the next time it is activated. In particular, every super-root in $A(C)$ will enter the root state the next time it is activated.
\end{lemma}
\begin{proof}
%Consider the case in which the root $p$ of some tree in $A(C)$ is not a boundary particle, i.e., $p$ is a follower.
Let $p$ be a follower %root
 boundary particle. By definition %$p$ is a boundary particle and by Lemma~\ref{lem:Trees}
 $p$ must have a retired particle or the object in its neighborhood. Therefore, $p$ immediately
becomes a root particle once it is activated according to
Algorithm~\ref{alg:spanningForestAlgorithm}.
\end{proof}

Furthermore, the following lemma provides a relation between the movement of
% ZD:phraseChange
%roots
 super-roots and the availability of complaint flags.

% ZD:phraseChange: in the following lemma
\begin{lemma}
\label{lem:ComplaintLeadsToMovement} For every tree of $A(C)$ with a
contracted
%root
 super-root $p$ and at least one complaint flag, $p$ will eventually
retire or expand to $p.dir$, thereby consuming a complaint flag, and after the
expansion $p$ may cease to be a %root
 super-root.
\end{lemma}
\begin{proof}
If $p$ is not a root, it becomes one the next time it is activated according to Lemma~\ref{lem:followerRootToBoundary}.
Therefore, assume $p$ is a root.
If there is a retired particle in $p.dir$, $p$ retires
and ceases to be a %root
 super-root. If the node in $p.dir$ is unoccupied, $p$ can
potentially expand. According to Algorithm~\ref{alg:complaint}, complaint
flags are forwarded along the tree of $p$ towards $p$. Once the flag reaches
$p$, it will expand, thereby consuming the flag. If $p$ expands, it might
have an active particle in its movement direction and thus ceases to be a
%root.
 super-root.
\end{proof}

Next, we prove the statement that expanded particles will not starve, i.e.,
they will eventually contract.

\begin{lemma}
\label{lem:particlesContract}
Eventually, every expanded particle contracts.
\end{lemma}
\begin{proof}
Consider an expanded particle $p$ in a configuration $C$. By
Lemma~\ref{lem:idleToActive} we can assume w.l.o.g. that all particles in $C$
are active or retired. If there is no particle $q$ with either $q.parent=p$ or
$p$ occupying the node in $q.dir$, then $p$
can contract once it is activated.
If such a $q$ exists and it is contracted, $p$ contracts in a handover (see Algorithm~\ref{alg:handover}).
If $q$ exists and is expanded, we consider the tree
of $A(C)$ that $p$ is part of. Consider a subpath in this tree that starts in
$p$, i.e., $(v_1,v_2,\ldots,v_k)$ such that $v_1,v_2$ are occupied by $p$ and
$v_k$ is a node that does not have an incoming edge in $A(C)$. Let $v_i$ be
the first node of this path that is occupied by a contracted particle. If all
particles are expanded, then clearly the last particle occupying $v_{k-1},v_k$
eventually contracts and we can set $v_i$ to $v_{k-1}$. Since $v_i$ is
contracted it eventually performs a handover with the particle occupying
$v_{i-2},v_{i-1}$. Now we can move backwards along $(v_1,v_2,\ldots,v_{i-1})$
and it is guaranteed that a contracted particle eventually performs a handover
with the expanded particle occupying the two nodes before it on the path. So
eventually $q$ is contracted, eventually performs a handover with $p$ and the
statement holds.
\end{proof}

In the following two lemmas we will specifically consider the case that $B \le
n$, i.e., the particles can coat at least one layer around the object.

\begin{lemma}
\label{lem:contractedLayer1}
If $B \le n$, layer $1$ is completely filled with contracted particles eventually.
\end{lemma}
\begin{proof}
Consider a configuration $C$ such that layer $1$ is not completely filled by
contracted particles. Note that in this case the leader election cannot have
succeeded yet, which means that a leader cannot be elected, and therefore
no particle can be retired in configuration $C$. So by
Lemma~\ref{lem:idleToActive} we can assume w.l.o.g. that all particles in
configuration $C$ are active.

Since layer 1 is not completely filled by contracted particles, there is either
at least one unoccupied node $v$ on layer $1$ or all nodes are occupied,
but there is at least one expanded particle on layer $1$. We show that in both
cases a follower will move to layer $1$, thereby filling up the layer until
all particles are contracted. In the first case, let $p$ be the
% ZD:phraseChange
%root
 super-root of a tree
in $A(C)$ that still has non-boundary particles, let $(p_0=p,p_1,\ldots,p_k)$ be
the boundary particles of the tree such that $p_{i-1}$ occupies the
node in $p_i.dir$ and let $q$ be the  non-boundary particle in the tree that is adjacent to
some $p_j in (p_0,\ldots,p_k)$ such that $j$ is minimal. If a particle $p_i$ in
$(p_0,\ldots,p_j=q.parent)$ is expanded, it eventually contracts
(Lemma~\ref{lem:particlesContract}) by a handover with $p_{i+1}$, and by
consecutive handovers all particles in $(p_{i+1},\ldots,p_j)$ eventually
expand and contract until the particle $p_j=q.parent$ expands. According to
Algorithm~\ref{alg:handover}, $p_j$ performs a handover with $q$. Therefore,
the number of particles on layer $1$ has increased. If all particles in
$(p_0,\ldots,q.parent)$ are contracted, then by
Lemma~\ref{lem:complaintTokenEquality} a complaint flag still exists in the
tree. Eventually, $p$ expands by Lemma~\ref{lem:ComplaintLeadsToMovement}.
Consequently, we are back in the former case that a particle in
$(p_0,\ldots,q.parent)$ is expanded.

In the second case, let $p'$ be an expanded boundary particle and let $q'$ be the
non-boundary particle with the shortest path in $A(C)$ to $p'$.
By a similar argument as for the first case, particles on layer $1$ perform
handovers (starting with $p'$) until eventually the node in $q'.parent$ is
occupied by a tail. Again, $q'$ eventually performs a handover and the number
of particles on layer $1$ has increased.
\end{proof}

As a direct consequence, we can show the following.

\begin{lemma}
\label{lem:leaderOnLayer1}
If $B \le n$, a leader is elected in layer $1$ eventually.
\end{lemma}
\begin{proof}
According to Lemma~\ref{lem:contractedLayer1} layer $1$ is eventually filled with contracted particles.
Leader Election successfully elects a leader node according to~\cite{DNA}.
The contracted particle $p$ occupying the leader node forwards the {\em CHK} flag and eventually receives it back, since all particles are contracted.
Therefore, $p$ becomes a leader.
\end{proof}

Now we are ready to prove the two major statements of this subsection that define two conditions for system progress.

\begin{lemma} \label{lem:liveness1}
If all particles are non-retired and there is either a complaint flag or an
expanded particle, the system eventually makes progress.
\end{lemma}
\begin{proof}
If there is an idle particle, progress is ensured by
Lemma~\ref{lem:idleToActive}. If an active particle is expanded
Lemma~\ref{lem:particlesContract} guarantees progress. Finally, in the last
case all particles are active, none of them is expanded and there is a
complaint flag. If layer $1$ is completely filled, a leader is elected
according to Lemma~\ref{lem:leaderOnLayer1} and as a direct consequence the
active particles on layer $1$ eventually retire, guaranteeing progress. If
layer $1$ is not completely filled, there exists at least one tree of $A(C)$
with a contracted
% ZD:phraseChange
%root
 super-root $p$ that has an unoccupied node in $p.dir$ and at least
one complaint flag. Therefore, progress is ensured by
Lemma~\ref{lem:ComplaintLeadsToMovement}.
\end{proof}

\begin{lemma} \label{lem:liveness2}
If there is at least one retired particle and one active particle, the system
eventually makes progress.
\end{lemma}
\begin{proof}
Again, if there is an idle particle, progress is ensured by
Lemma~\ref{lem:idleToActive}. Moreover, note that since there is at least one
retired particle, we can conclude that leader election has been successful
(since the first particle that retires is a leader particle) and therefore layer $1$ has to be
completely filled with contracted particles. If there is still a non-retired
particle on layer $1$, it eventually retires according to the Algorithm,
guaranteeing progress.

So suppose that all particles in layer 1 are retired. We distinguish between
the following cases: (i) there exists at least one
% ZD:phraseChange
%root
 super-root, (ii) no
% ZD:phraseChange
%root
super-root exists,
but there is an expanded particle, and (iii) no
% ZD:phraseChange
%root
 super-root exists and all particles
are contracted. In case (i), Lemma~\ref{lem:followerRootToBoundary} guarantees
that a
% ZD:phraseChange
%root
 super-root will eventually enter %guide
 root state, and therefore it will eventually either expand (if $p.dir$ is unoccupied) or retire  (since $p.dir$ is
occupied by a retired particle). In case (ii), the particle contracts
according to Lemma~\ref{lem:particlesContract}. In case (iii) $A(C)$ forms a
ring of trees, which can only happen if all boundary particles
completely occupy a single layer, so there is an active particle that occupies
the node adjacent to the marker edge. Since it is contracted by assumption, it retires upon
activation.
%Otherwise, Lemma~\ref{lem:particlesContract} that eventually
%either that particle or the particle it eventually performs a handover with
%occupies the node adjacent to the marker edge in contracted form, which means
%that upon activation it retires.
Therefore, in all three cases the system
eventually makes progress.
\end{proof}

\subsection{Termination}

Finally, we show that the algorithm eventually terminates in a legal
configuration, i.e., a configuration in which the coating problem is solved.
For the termination we need the following two lemmas.

\begin{lemma} \label{lem:progress-1}
The number of times an idle particle is transformed into an active one and an
active particle is transformed into a retired one is bounded by
$\mathcal{O}(n)$.
\end{lemma}
\begin{proof}
From our algorithm it immediately follows that every idle particle can only be
transformed once into an active particle, and every active particle can only
be transformed once into a retired particle. Moreover, a non-idle particle can
never become idle again, and a retired particle can never become non-retired
again, which proves the lemma.
\end{proof}

\begin{lemma} \label{lem:progress-2}
The overall number of expansions, handovers, and contractions is bounded by
$\mathcal{O}(n^2)$.
\end{lemma}
\begin{proof}
We will need the following fact, which immediately follows from our algorithm.

\begin{fact} \label{fact-1}
Only a
% ZD:phraseChange root
 super-root of $A(C)$ can expand to a non-occupied node, and every such
expansion triggers a sequence of handovers, followed by a contraction, in
which every particle participates at most twice.
\end{fact}

Consider any particle $p$. Note that only an active particle performs a
movement. Let $C$ be the first configuration in which $p$ becomes active. If
it is a non-boundary particle (i.e., a follower), then consider the directed path in $A(C)$ from the head of
$p$ to the
% ZD:phraseChange
%root
 super-root $r$ of its tree or the first particle $r$ belonging to the
ring in the ring of trees. Such a path must exist due to
Lemma~\ref{lem:Trees}. Let $P= (v_0, v_1, \ldots, v_m)$ be a node sequence
covered by this path where $v_0$ is the head of $p$ in $C$ and $v_m$ is the
first node along that path with the object or a retired particle in its
neighborhood. Note that by Lemma~\ref{lem:Trees} such a node sequence is
well-defined since $v_m$ must at latest be a node occupied by $r$. According to
Algorithm~\ref{alg:spanningForestAlgorithm}, $p$ attempts to follow $P$ by
sequentially expanding into the nodes $v_0, v_1, \ldots, v_m$. At latest, $p$
will become a boundary particle once it reaches $v_m$. Up to this point, $p$
has traveled along a path of length at most $2n$, and therefore, the number of
movements $p$ executes as a follower is $\mathcal{O}(n)$.

Now suppose $p$ is a boundary particle. Let $C$ be the configuration in which
$p$ becomes a boundary particle and let $\ell=p.layer$. Suppose that $\ell=1$.
From our algorithm we know that at most $n$ complaint flags are generated by
the particles, and therefore by Lemma~\ref{lem:ComplaintLeadsToMovement},
there are at most $n$ expansions in level 1 (the rest are handovers or
contractions). Hence, it follows from Fact~\ref{fact-1} that $p$ can only move
$\mathcal{O}(n)$ times as a boundary particle.

Next consider the case that $\ell>1$. Here we will need the following
well-known fact.

\begin{fact} \label{fact-2}
Let $B_i$ be the length of layer $i$. For every $i$ and every valid instance
$(P,O)$ allowing $O$ to be coated by $i$ layers it holds that $B_i = B_0 +
6i$.
\end{fact}

If $\ell=2$, there must be a retired particle in layer 1, and since the leader
is the first retired particle, Lemmas~\ref{lem:contractedLayer1} and
\ref{lem:leaderOnLayer1} imply that level $\ell-1$ is completely filled with
contracted particles. So $p$ can only move along nodes of layer $\ell$. Since
$B_{\ell-1} \le n$, it follows from Fact~\ref{fact-2} that $B_{\ell} \le n+6$.
As long as not all particles in level $\ell-1$ are retired, $p$ cannot move
beyond the marker node in level $\ell$. So $p$ either becomes retired before
reaching the marker node, or if it reaches the marker node, it has to wait
there till all particles in level $\ell-1$ are retired, which causes the
retirement of $p$. Therefore, $p$ moves along at most $n+6$ nodes. If
$\ell>2$, we know from Lemma~\ref{lem:NoHole} that level $\ell-2$ is
completely filled with contracted particles. Since $B_{\ell-2} \le n$ and
$B_{\ell} = B_{\ell-2} + 12$, it follows that $B_{\ell} \le n+12$. Hence, $p$
will move along at most $n+12$ nodes in level $\ell$ before becoming retired
or moving to level $\ell-1$, and $p$ will move along at most $n+6$ further
nodes in level $\ell-1$ before retiring.

Thus, in any case, $p$ performs at most $\mathcal{O}(n)$ movements as a
boundary particle. Therefore, the number of movements any particle in the
system performs is $\mathcal{O}(n)$, which concludes the lemma.
\end{proof}

Lemmas~\ref{lem:progress-1} and \ref{lem:progress-2} imply that the system can
only make progress $\mathcal{O}(n^2)$ many times. Hence, eventually our system
reaches a configuration in which it no longer makes progress, so the system
terminates. It remains to show that when the algorithm terminates, it is in a
legal configuration, i.e., the algorithm solves the coating problem.

\begin{theorem}
Our coating algorithm terminates in a legal configuration.
\end{theorem}
\begin{proof}
From the conditions of Lemmas~\ref{lem:liveness1} and \ref{lem:liveness2} we
know that the following facts must both be true when the algorithm terminates:
\begin{enumerate}
\item At least one particle is retired or there is neither a complaint flag nor
an expanded particle in the system (Lemma~\ref{lem:liveness1}).

\item Either all particles are retired or all particles are active
(Lemma~\ref{lem:liveness2}).
\end{enumerate}
First suppose that all particle are retired. Then it follows from
Lemma~\ref{lem:NoHole} that the configuration is legal. Next, suppose that all
particles are active and neither a complaint flag nor an expanded particle is
left in the system. Then Lemma~\ref{lem:complaintTokenEquality} implies that
there cannot be any non-boundary any more, so all active particles must be
boundary particles. If there is at least one boundary particle in layer
$\ell>1$, then there must be at least one retired particle, contradicting our
assumption. So all boundary particles must be in layer 1, and since there are
no more complaint flags and all boundary particles are contracted, also in
this case our algorithm has reached a legal configuration, which proves our
theorem.
\end{proof}

Recall that the {\em work} performed by an algorithm is defined as the number
of movements (expansions, handovers, and contractions) of the particles till
it terminates. Lemma~\ref{lem:progress-2} implies that the work performed by
our algorithm is $\mathcal{O}(n^2)$. Interestingly, this is also the best
bound one can achieve in the worst-case for the coating problem.

\begin{lemma} \label{lem:worstCase}
The worst-case work required by any algorithm to solve the Universal Object
Coating problem is $\Omega(n^2)$.
\end{lemma}
\begin{proof}
Consider the configuration depicted in Figure~\ref{fig:worstCase}.
    \begin{figure}
        \centering
        \includegraphics{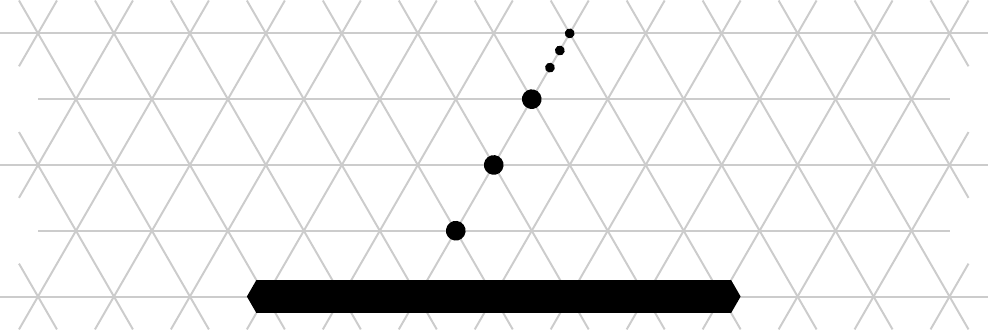}
        \caption{A worst-case configuration concerning work.
            The object is solid black and the non-object particles are black dots .
            Here, all $n$ particles lie on a straight line.
        }
        \label{fig:worstCase}
    \end{figure}
A particle with distance $i\geq 1$ to the object needs at least
$2(i-1-\left\lceil\frac{i-1}{B}\right\rceil)$ movements to become contracted
on its final layer. Therefore, any algorithm requires at least
$2\sum_{i=1}^{n-1} (i-1-\left\lceil \frac{i-1}{B}\right\rceil ) \geq
\sum_{i=1}^{n-1} (i-1-(\frac{i}{B}))=\Omega(n^2)$ work assuming $B \geq 2$.
\end{proof}

Hence, we get:

\begin{theorem}
Our algorithm requires worst-case optimal work $\Theta(n^2)$.
\end{theorem}

\section{Applications}
\label{sec:applications}
In this section, we present other coating scenarios and applications of our universal coating algorithm.
Our algorithm can be easily extended to also handle the case when one would like to cover only a certain portion of the object surface.  More concretely, assume that one would like to cover the  portion of the object surface delimited by two \emph{endpoint nodes}. Basically in that case, the algorithm can be modified slightly so that the particles that eventually reach one of the endpoints of the surface segment retire and become {\em endpoint markers}. The position of endpoint marker particles will be propagated to higher layers, as necessary, such that the delimited portion of the object is evenly coated.

%The endpoint marker particles would act in the same way as the marker particles in our original algorithm: The only modification needed would be that for a new endpoint marker particle $p$ to emerge at layer $l$, $p$ would simply check if it is adjacent to an endpoint marker particle $p'$ at layer $l-1$ and whether $p'.CCW$ (resp., $p'.CW$)  is retired if $l$ is odd (resp., even).

%\textbf{RG: Neither Thim nor I really understood this description. I would suggest to remain on a higher abstraction level: The position of marker particles can be propagated upwards such that the delimited portion of the object is evenly coated.}

Once the first layer is formed and a leader is elected (implying that $B\leq n$), one can trivially determine $(i)$ whether the number of particles in the system is greater than or equal to the size of the object boundary, or $(ii)$ whether the object $O$ is convex; one could also potentially address other applications that involve aggregating some (constant-size) collective data over the boundary of the object $O$.
Once all particles in layer 1 retire, a leader will emerge
and that leader can initiate the respective application.
For the first application, all particles may initially assume that $B>n$. Once a leader is elected, it informs all other particles that $B\leq n$. For the convexity testing,
the leader particle can generate a token that traverses the boundary in CW direction: If the token ever makes a left turn (i.e., it traverses two consecutive edges on the boundary at an outer angle of less than $180^{\circ}$), then the object is not convex; otherwise the object is convex.

%One could also use the universal coating algorithm per se to identify if the object contains narrow tunnels.
%Here we present a simple rule that would allow us to identify odd width narrow tunnels: If a particle $p$ in layer $l$ has a set of adjacent  retired particles in layer $l-1$ that are not connected, then $p$ reports a tunnel of width less than $2\lceil n/B \rceil$  on the object surface. However, detecting even width narrow tunnels turns out to be a bit more involved.

%{\bf AR: add link to videos on website. add applications to abstract as well? clean up this section.}

\section{Conclusion}
\label{sec:concl}

This paper presented a universal coating algorithm for programmable
matter using worst-case optimal work. It would be interesting to also bound the
parallel runtime of our algorithm in terms of number of asynchronous rounds, and to investigate its competitiveness ---
i.e., how does its work or runtime compare to the best possible work or
runtime for any given instance. Moreover, it would be interesting to implement the
algorithm and evaluate its performance either via simulations or hopefully at
some point even via experiments with real programmable matter.

\section*{Acknowledgements}
We would like to thank Joshua Daymude and Alexandra M.\ Porter for fruitful discussions on this topic and for helping us review the manuscript.
%\newpage
%\thispagestyle{empty}
\bibliographystyle{elsarticle-num}
\bibliography{literature}

\end{document}